\documentclass[10pt, a4paper]{article}
 
\usepackage{bm}

\usepackage[utf8]{inputenc}
\usepackage[T1]{fontenc}

\usepackage{amsmath,amssymb}
\usepackage{mathtools}  
\usepackage{amsthm} 
\usepackage{mathrsfs} 
\usepackage{bbm} 

\usepackage{hyperref}

\theoremstyle{plain}
\newtheorem{theorem}{Theorem}[section]
\newtheorem*{theorem*}{Theorem} 
\newtheorem{proposition}[theorem]{Proposition}
\newtheorem{lemma}[theorem]{Lemma}
\newtheorem{corollary}{Corollary}[theorem]

\theoremstyle{definition}
\newtheorem{definition}[theorem]{Definition}
\theoremstyle{definition}

\theoremstyle{definition}
\newtheorem*{definition*}{Definition}

\theoremstyle{definition}
\newtheorem{remark}[theorem]{Remark}
\theoremstyle{definition}
\newtheorem*{remark*}{Remark}
\theoremstyle{definition}

\theoremstyle{definition}
\newtheorem*{remarks*}{Remarks}

\theoremstyle{remark}

\theoremstyle{remark}

\newcommand{\deq}{\stackrel{\mathrm{def}}{=}}
\newcommand{\dd}{\,\mathrm{d}}
\newcommand{\ii}{\mathrm{i}}
\newcommand{\id}{\,\boldsymbol{1}}

\DeclareMathOperator*{\Dom}{Dom}
\DeclareMathOperator*{\Ran}{Range}
\newcommand{\PP}{\mathbb P}
\newcommand{\NN}{\mathbb N}

\newcommand{\RR}{\mathbb R}
\newcommand{\CC}{\mathbb C}

\newcommand*{\CCliff}{\CC\ell}

\DeclareMathOperator*{\End}{End}

\newcommand{\Expect}{\mathbb{E}}

\makeatletter
\newcommand{\extp}{\@ifnextchar^\@extp{\@extp^{\,}}}
\def\@extp^#1{\mathop{\bigwedge\nolimits^{\!#1}}}
\makeatother

\newcommand*{\Extern}{\extp\!}

\newcommand*{\Spin}{\mathrm{Spin}}
\newcommand*{\SO}{\mathrm{SO}}

\newcommand{\tmaffiliation}[1]{\\ #1}

\title{Finite dimensional systems of free Fermions and diffusion processes on $\Spin$ groups}
\author{L. M. Borasi}

\author{Luigi M. Borasi\\[.5em]
  \tmaffiliation{Hausdorff Center of Mathematics \&\\
  Institute of Applied Mathematics\\
  University of Bonn, Germany}\\[.2em]
  \tt{\small borasi@iam.uni-bonn.de}
}

\begin{document}
\maketitle

\begin{abstract}
  In this article we are concerned with ``finite dimensional Fermions'',
  by which we mean vectors in a finite dimensional complex space
  embedded in the exterior algebra over itself.  
  These Fermions are spinless but possess the characterizing anticommutativity property.
  We associate invariant complex vector fields on the Lie group $\Spin(2n+1)$
  to the Fermionic creation and annihilation operators. 
  These vector fields are elements of the complexification of the regular representation of the Lie algebra
  $\mathfrak{so}(2n+1)$.
  As such, they do not satisfy the canonical anticommutation relations,
  however, once they have been projected onto an appropriate subspace of $L^2(\Spin(2n+1))$,
  these relations are satisfied.
  We define a free time evolution of this system of Fermions in terms of a symmetric positive-definite
  quadratic form in the creation-annihilation operators.
  The realization of Fermionic creation and annihilation operators brought by
  the (invariant) vector fields allows us to interpret this time evolution
  in terms of a positive selfadjoint operator which is the sum of a second order operator,
  which generates a stochastic diffusion process, and a first order
  complex operator, which strongly commutes with the second order operator.
  A probabilistic interpretation is given in terms of a Feynman-Kac like formula
  with respect to the diffusion process associated
  with the second order operator. 
\end{abstract}

\maketitle

\section{Introduction}
Probabilistic methods in Quantum Field Theory have proved to be particularly fruitful
(cf.\ e.g.~\cite{simon_pphi2_1974, glimm_quantum_1987, hida_white_1993}).
These methods have been almost exclusively restricted 
to \textit{Bosonic} Field Theories.
Some ideas of the Bosonic probabilistic methods
carry over, to an extent, to the Fermionic case using the beautiful algebraic technique of
Berezin integration~\cite{berezin_method_1966}. 
However, the Berezin integral, being defined in terms of Grassmann variables, 
does not lend itself easily to an interpretation in the context of probability theory or measure theory
(nevertheless, at least for the case of a discrete number of variables,
a probabilistic interpretation of the Berezin formalism is possible, albeit somewhat cumbersome:
cf.\ e.g.~\cite{combe_poisson_1980,de_angelis_berezin_1998}).

The aim of this work is to study a finite dimensional system of free Fermionic creation-annihilation
operators in a way which parallels, in the sense explained below,
the treatment of the corresponding Bosonic case of a
finite dimensional quantum harmonic oscillator.
It is well known that the Hamiltonian of an $n$-dimensional quantum harmonic oscillator
can be interpreted, in Euclidean times, as the second order differential operator
which generates an Ornstein-Uhlenbeck process on $\RR^n$
(cf.\ e.g.~\cite[p. 35]{simon_functional_2005}).
Following this Bosonic parallel we study a model for a finite dimensional system of Fermions
where the Fermionic Hamiltonian is replaced, in a quite natural way, by a second order differential operator.
Moreover we study the possibility of interpreting the time evolution generated by such Hamiltonian
in terms of stochastic processes.

The results in this work are inspired in part by the work of Schulman~\cite{schulman_path_1968} 
(cf.\ also~\cite[Chapters 22-24]{schulman_techniques_2005}) 
who gives a description of a single $\tfrac12$-spin particle 
in terms of the Feynman path integral. 
The precedents for Schulman's idea can be found in early work on Quantum Mechanics
connecting the Pauli $\tfrac12$-spin formalism with the
\textit{quantum spinning top} \cite{bopp_uber_1950,rosen_particle_1951} (cf. also the more recent work \cite{barut_magnetic_1992}).
The works which, to the knowledge of the author, are the closest in spirit to the analysis given here are 
\cite{fukutome_new_1977} and \cite{fukutome_so2n+1_1977}.
The situation studied here is nevertheless quite different.
We are motivated by the parallel between the (Bosonic) quantum harmonic oscillator and
the Ornstein-Uhlenbeck process.
For this reason we associate to the Hamiltonian of a finite system of free Fermions
a second order operator whereas in the works cited above the Hamiltonian was associated
to a first order operator.
The analysis that follows is therefore completely different.
Moreover, because of our motivation, we pay particular attention
to the rigorous, functional analytic details of our description.

We start our analysis by considering
Fermionic creation-annihilation operators $a_j^\dagger,a_j$, $j=1,\dots,n$, $n\in\NN$.
We associate to these Fermionic operators
first order differential operators.
We achieve this purpose, similarly to the works cited above, by first exploiting the standard fact
that the Fermionic creation-annihilation operators
give rise to a faithful, irreducible representation of the complex Lie algebra
$\mathfrak{so}(2n+1,\CC)$ (obtained by complexifying
the real Lie algebra $\mathfrak{so}(2n+1)$ of antisymmetric real $(2n+1)\times(2n+1)$-matrices).
This representation is usually called the \textit{spin representation} of $\mathfrak{so}(2n+1,\CC)$.
We therefore associate the Fermionic creation-annihilation operators
with abstract elements in the Lie algebra
$\mathfrak{so}(2n+1,\CC)$.
We then use the standard fact that this Lie algebra can be realized in terms
of (left-invariant) differential operators acting on smooth functions from the (real) Lie group $\Spin(2n+1)$
to $\CC$.
Here $\Spin(2n+1)$ is the simply-connected Lie group obtained as universal cover of
the Lie group $\SO(2n+1)$ of rotations in $2n+1$ dimensions.
We need to pass from $\SO(2n+1)$ to $\Spin(2n+1)$
because the spin-representation of $\mathfrak{so}(2n+1,\CC)$
does not appear inside the representation of $\mathfrak{so}(2n+1,\CC)$ given by left-invariant
vector fields acting on $C^\infty(\SO(2n+1))$ (indeed the spinor representation of $\SO(2n+1)$
is only a projective representation of $\SO(2n+1)$ and not an actual representation).
Having associated the Fermionic creation-annihilation operators $a_j^\dagger,a_j$
to differential operators $D_j^+,D_j^-$ 
we consider the free Fermionic Hamiltonian $H=\sum_{j=1}^n E_j a_j^\dagger a_j$
(for positive constants $E_j$) and we lift it to an element $\tilde{H}$ of the universal enveloping algebra
of $\mathfrak{so}(2n+1,\CC)$ which we look upon as the algebra of differential operators on $C^\infty(\Spin(2n+1)$
generated by the left-invariant vector fields together with the identity.
The lift $H\leadsto \tilde{H}$ is by its very nature non-canonical.
We choose to define $\tilde{H}=\sum_{j=1}^n D_j^+D_j^-$
by formally replacing the creation-annihilation operators in $H$
by their associated first order differential operators.
This choice differs from the one made by the references cited above.
The motivation for our choice is that we want to study a situation parallel to the Bosonic case,
where the free Hamiltonian of a quantum harmonic oscillator corresponds to a second order differential operator.
The main results we obtain about the operator $\tilde{H}$ are contained in theorem
\ref{Theorem:quasi Hamiltonian}, which contains functional analytic properties regarding the operator
$\tilde{H}$, and in theorem \ref{Feynam-Kac},
where we give a Feynman-Kac like formula describing the evolution in Euclidean time
generated by $-\tilde{H}$.

The layout of the article is as follows.
In section \ref{2} we give some basic definitions and describe the standard relation between $n$
Fermionic creation-annihilation operators and the spin representation of the Lie algebra
$\mathfrak{so}(2n+1,\CC)$, in a way which is well suited for our needs.
In section \ref{3} we briefly describe the standard connection between Lie algebra
and left-invariant differential operators.
We then specialize this general relation to our setting and describe
how to recover, from this global picture, the spin representation of $\mathfrak{so}(2n+1,\CC)$.
In section \ref{4}, after some remarks regarding selfadjointness
and the universal enveloping algebra of a compact Lie group,
we define the operator $\tilde{H}$ and describe some of its most salient functional analytic properties
(theorem \ref{Theorem:quasi Hamiltonian}).
Finally in section \ref{5} we introduce some standard facts regarding stochastic processes on Lie groups
and we apply the general theory to our case.
The main result is that the operator $\tilde H$
splits into two parts, a strictly second order part which is hypoelliptic and generates a diffusion on $\Spin(2n+1)$
and a first order part which, as explained in section \ref{5},
does not contribute to a diffusion on $\Spin(2n+1)$.
On the other hand, since it strongly commutes with the second order part,
we are able to write a simple Feynman-Kac like formula for the operator $\tilde{H}$
where we average over the process generated by the strictly second order part of $-\tilde{H}$.
We note that our Feynman-Kac like formula resembles the
Feynman-Kac formula for the Schr\"odinger operator of a particle in a magnetic field
(cf. e.g. \cite[Chapter V, Section 15.]{simon_functional_2005}).

\section{Fermions and \texorpdfstring{$\mathfrak{so}(2n+1,\CC)$}{so(2n+1,C)}}\label{2}
In this section, after some notational preliminaries,
we describe the relation between the Fermionic creation-annihilation operators and
the $1/2$-spin representations of the Lie algebra $\mathfrak{so}_\CC$.
Since this relation is not very broadly known we provide some details.

Let us denote by $\CCliff(N)$, $N\in\NN$, the complex \textit{Clifford algebra} over $\CC^N$,
that is the unital associative algebra obtained as the quotient
of the full tensor algebra $T(\CC^N)$ by the following relations, 
\begin{equation*}
  \{v,w\} = -2\langle v,w\rangle\id,
  \quad v,w\in\CC^N
  ,
\end{equation*}
where $\langle\cdot,\cdot\rangle$
denotes the standard symmetric bilinear form on $\CC^n$,
and $\{v,w\}\deq vw+wv$ where we have denoted the product of $v,w$ in $\CCliff(N)$
simply by $vw$.

Let $V$ be a real or complex , finite dimensional, vector space.
We denote by $\Extern V$ the \textit{exterior algebra} over $V$, that is the algebra obtained by quotienting
the full tensor algebra $T(V)$ over $V$ by the two sided ideal generated by elements of the form
$v\otimes v$, $v\in V$. 
The finite dimensional Fermionic Fock space $\Gamma_\wedge\CC^n\deq\bigoplus_{k=0}^n (\CC^n)^{\wedge k}$,
$n\in\NN$,
is defined as the Hilbert space realized by taking
the exterior algebra $\Extern \CC^n$ just as vector space and equipping it with the
Hermitian scalar product $(\cdot,\cdot)_{\Gamma\CC^n}$ which satisfies
$(v_1\wedge\cdots\wedge v_n,w_1\wedge\cdots\wedge w_n)_{\Gamma_\wedge\CC^n} \deq \det_{jk}(v_j,w_k)_{\CC^n}$
for $v_j,w_k\in\CC^n$, $j,k=1,\dots,n$,
where $(\cdot,\cdot)_{\CC^n}$ denotes the standard Hermitian scalar product on $\CC^n$
(antilinear in the left component).

We call the element $1\in\CC\hookrightarrow\Extern\CC^n$ the 
\textit{vacuum vector}.
For $v\in\CC^n$, we denote the \textit{Fermionic creation}, \textit{annihilation operators} on $\Gamma_\wedge\CC^n$ 
respectively by $c^\dagger(v)$, $c(v)$.
Explicitly, for any $v\in\CC^n$, $c(v)$ and $c^\dagger(v)$ are defined as Hilbert adjoint of each other in $\Gamma_\wedge\CC^n$
with $c^\dagger(v)\psi = v\wedge\psi$ where $\psi\in\Gamma_\wedge\CC^n$.
Note that by this definition $c^\dagger(v)$ is (complex) linear in $v$
whereas $c(v)$ is (complex) anti-linear.
If $e_j$, $j\in\{1,\dots,n\}$, denotes the standard basis of $\CC^n$,
then we denote $c(e_j)$, respectively $c(e_j)^\dagger$, by $c_j$, respectively $c_j^\dagger$.
They satisfy the usual canonical anticommutation relations: $\{c_j,c_k^\dagger\}=\delta_{jk}$,
$\{c_j,c_k\}=\{c_j^\dagger,c_k^\dagger\}=0$
(where $\{A,B\}\deq AB+BA$ for any $A,B$ in an associative algebra).

Let us fix a decomposition $\CC^{2n}=\CC^n\oplus\CC^n$
orthogonal with respect to the standard symmetric bilinear form $\langle\cdot,\cdot\rangle$
of $\CC^{2n}$, and let $\mathbb P_1$,
respectively $\mathbb P_2$, be the projection of $\CC^{2n}$ onto the first,
respectively the second, copy of $\CC^n$.
Moreover let us denote by $\bar v$ the vector obtained from $v\in\CC^{2n}$
by complex conjugating each component.
Then the standard Hermitian scalar product $(\cdot,\cdot)$ on $\CC^{2n}$ satisfies
$(v,w)=\langle\bar v, w\rangle$, $v,w\in\CC^{2n}$.
Let us define the algebra isomorphism $\gamma:\CCliff(2n)\rightarrow \End(\Extern\CC^n)$
by extending to the whole of $\CCliff(2n)$
the relations
\begin{equation}
  \label{Fock rep}
  \begin{aligned}
    \gamma(v) &= c^\dagger(\mathbb P_1 v) - c(\overline{\mathbb P_1 v})
    -\ii ( c^\dagger((\mathbb P_2 v) + c(\overline{\mathbb P_2 v}) ),
    \quad v\in\CC^{2n}
    .
  \end{aligned}
\end{equation}
The map $\gamma$ is a representation
of $\CCliff(2n)$ on the Fock space $\Gamma_\wedge\CC^n$
usually called the \textit{Fock space representation} of $\CCliff(2n)$.

Let $e_j$, $j=1,\dots,2n$, be the standard basis of $\CC^{2n}$.
We also denote by $e_\ell$, $\ell=1,\dots,n$, a basis for each of the copies of $\CC^n$
in the decomposition $\CC^{2n}=\CC^n\oplus\CC^n$.
To be concrete, in the following we will take $\PP_1$, to be  the projection which sends $e_{2j-1}$ to $e_j$ and $e_{2j}$ to zero.
Then $\PP_2$ sends $e_{2j}$ into $e_j$ and $e_{2j-1}$ to zero.

Let $\gamma_j \deq \gamma(e_j)$, $j=1,\dots,2n$.
Then, with this choice of $\PP_1$, $\PP_2$ we obtain, from \eqref{Fock rep}, the following relations
\begin{equation}
  \label{creation annihilation j}
  \begin{aligned}
    c_j^\dagger &= \tfrac12(\gamma_{2j-1} + \ii\gamma_{2j})\\
    c_j &= \tfrac12(-\gamma_{2j-1} + \ii\gamma_{2j}),
    \quad j=1,\dots n
    .
  \end{aligned}
\end{equation}
Note that, under these definitions, the operators $\gamma_j=\gamma(e_j)$,
are anti-Hermitian
as operators on the finite dimensional Hilbert space $\Gamma_\wedge\CC^n$.

Let $\mathfrak{so}(N)$,
$N\in\NN$, denote the complex Lie algebra of antisymmetric $N\times N$ real matrices.
and let $\mathfrak{so}(N,\CC)=\mathfrak{so}(N)\otimes_\RR\CC$ be its complexification.

We now define the standard $1/2$-spin representation of $\mathfrak{so}(2n+1,\CC)$
on $\Extern\CC^n$. 
We give the definition in a form which differs slightly from the standard presentations
(cf. e.g. \cite{fulton_representation_1991,goodman_symmetry_2009}) therefore we provide some details.

Consider an embedding $\iota:\mathfrak{so}(2n)\hookrightarrow\mathfrak{so}(2n+1)$
and the relative vector space decomposition
$\mathfrak{so}(2n+1,\CC)=V_{2n} \oplus \iota(\mathfrak{so}(2n,\CC))$,
where $V_{2n}=\mathfrak{so}(2n+1,\CC)/\iota(\mathfrak{so}(2n,\CC))$ is
a $2n$-dimensional vector space which generates, via the Lie brackets, all of $\mathfrak{so}(2n+1,\CC)$.

Let $\CCliff(V_{2n})$ be the complex Clifford algebra
generated by the identity $\id$ and by the symbols
$\kappa(X)$, $X\in V_{2n}$, which satisfy
\begin{equation*}
  \{\kappa(X),\kappa(Y)\} = \tfrac{1}{4}\mathrm{tr}(XY)\id,\quad X,Y\in V_{2n},
\end{equation*}
where $\mathrm{tr}$ denotes the trace in the defining representation of $\mathfrak{so}(2n+1,\CC)$
and $XY$ denotes the product of $X$ and $Y$ as $(2n+1)\times(2n+1)$ matrices.
Let us identify $V_{2n}$ with $\CC^{2n}$ and denote one such isomorphism by $\phi$.
Then $\phi$ extends to an isomorphism of $\CCliff(V_{2n})$ with $\CCliff(2n)$
and the composition
$$
\gamma\circ\phi:\CCliff(V_{2n})\rightarrow\textrm{End}(\Extern\CC^n)
,
$$
of the isomorphism $\phi$ with a Fock representation $\gamma$ of $\CCliff(2n)$,
defines a Fock representation\footnote
{
  One could more efficiently define a Fock representation directly of $\CCliff(V_{2n})$
  without any reference to $\CCliff(2n)$, $\CC^{2n}$, or the non-canonical isomorphism $\phi$.
  We choose to introduce $\phi$ and replace $V_{2n}$ by $\CC^{2n}$ and $\CCliff(V_{2n})$ by $\CCliff(2n)$
  to make the computations in the following sections more concrete. 
} of $\CCliff(V_{2n})$.
The following proposition shows how the Clifford algebra $\CCliff(V_{2n})$
with a Fock representation $\gamma\circ\phi$
gives rise to an irreducible representation $\pi^{1/2}$ of $\mathfrak{so}(2n+1,\CC)$
which is unique up to isomorphism and coincides with the standard $1/2$-spin representation
of $\mathfrak{so}(2n+1)$.
\begin{proposition}
  The map $\kappa$ extends to a Lie algebra homomorphism
  $\mathfrak{so}(2n+1,\CC)\rightarrow\CCliff(V_{2n})$, still denoted by $\kappa$,
  which sends the Lie brackets of $\mathfrak{so}(2n+1,\CC)$ into the commutator $[A,B]=AB-BA$,
  for $A,B\in\CCliff(V_{2n})$.
  The composition
  \begin{equation*}
    \pi^{1/2}\deq \gamma\circ\phi\circ\kappa :\mathfrak{so}(2n+1,\CC)\rightarrow\text{End}(\Extern\CC^n)
    ,
  \end{equation*}
  of the homomorphism $\kappa$ with the Clifford algebra isomorphism $\phi$ and
  with the Clifford algebra representation $\gamma$,
  defines a representation $\pi^{1/2}$ of $\mathfrak{so}(2n+1,\CC)$.
  This representation is isomorphic to the standard
  $1/2$-\textit{spin representation} of $\mathfrak{so}(2n+1,\CC)$,
  that is, the irreducible representation of $\mathfrak{so}(2n+1,\CC)$
  on $\Extern\CC^n$.
\end{proposition}
\begin{proof}
  Without loss of generality let us fix $\{X_{jk}\}_{1\le j<k\le 2n+1}$,
  to be the standard basis of $\mathfrak{so}(2n+1,\CC)$
  given by  matrices $X_{jk}=e_j\wedge e_k\deq e_j\otimes e_k-e_k\otimes e_j$, where $e_\ell$, $\ell=1,\dots,2n+1$,
  is the standard basis of $\CC^{2n+1}$. Explicitly, these matrices have components 
  $[X_{jk}]_{j'k'}=\delta_{jj'}\delta_{kk'}-\delta_{jk'}\delta_{kj'}$,
  $1\le j<k\le 2n+1$, $j',k'\in\{1,\dots,2n+1\}$.
  and satisfy the commutation relations,
  for $1\le r<i\le 2n+1$ and $1\le s<j\le 2n+1$,
  \begin{equation}
    \label{Lie brakets}
    [X_{ri},X_{sj}] = \delta_{is} X_{rj} + \delta_{rj}X_{is}  -\delta_{ij}X_{rs}-\delta_{rs}X_{ij}
    .
  \end{equation}
  We consider the embedding $\iota:\mathfrak{so}(2n,\CC)\hookrightarrow\mathfrak{so}(2n+1,\CC)$
  obtained by setting
  $\iota(\mathfrak{so}(2n,\CC))$ to be the subalgebra of $\mathfrak{so}(2n+1,\CC)$
  spanned by $\{X_{jk}\}_{j,k\in\{1,\dots,2n\}}$,
  and we let $V_{2n}$ be the vector space spanned by $\{X_{j,2n+1}\}_{j=1,\dots,2n}$.
  Then we have the decomposition $\mathfrak{so}(2n+1,\CC) = V_{2n} \oplus \iota(\mathfrak{so}(2n,\CC))$,
  indeed,
  $V_{2n}\cap\iota(\mathfrak{so}(2n,\CC))= \{0\}$,
  and the direct sum of $V_{2n}$ and
  $\mathfrak{so}(2n,\CC)$ is spanned by the whole basis above hence coincides
  with $\mathfrak{so}(2n+1,\CC)$.

  A straightforward computation shows that
  $\textrm{tr}(X_{a,2n+1}X_{b,2n+1}) = -2\delta_{ab}$,
  $a,b\in\{1,\dots,n\}$.
  Hence, in this basis, we have 
  \begin{equation*}
    \{\kappa(X_{a,2n+1}),\kappa(X_{b,2n+1})\} = -\tfrac12\delta_{ab},
    \quad a,b\in\{1,\dots,2n\}
    .
  \end{equation*}
  Using these anticommutation relations and the commutation relations \eqref{Lie brakets}
  one can easily show that $\kappa$ extends to a homomorphism
  of Lie algebras as claimed in the statement of the proposition.

  Let us identify $V_{2n}$ with $\CC^{2n}$,
  and therefore $\CCliff(V_{2n})$ with $\CCliff(2n)$,
  by the map $\phi$ which sends, for all $j=1,\dots,2n$, $X_{j,2n+1}$
  into the standard basis element $e_j$ of $\CC^{2n}$.
  Then we have
  \begin{equation*}
    \begin{aligned}
      (\phi\circ\kappa)(X_{j,2n+1}) &= \tfrac12 e_j,\quad j=1,\dots,2n,\\
      (\phi\circ\kappa)(X_{jk}) &= \tfrac12 e_je_k,\quad 1\le j<k\le 2n
      .
    \end{aligned}
  \end{equation*}

  We let, as in \eqref{creation annihilation j}, $\gamma$ be the Fock space representation of $\CCliff(2n)$ 
  with projections $\PP_1, \PP_2$ such that $\PP_1e_{2j-1}= \PP_2e_{2j} = e_{j}$, $\PP_1e_{2j} = \PP_2e_{2j-1}=0$.  
  Now, since $\gamma$ is a representation of $\CCliff(2n)$ on $\Extern\CC^n$,
  it is clear that $\pi^{1/2}=\gamma\circ\phi\circ\kappa$ extends to a representation
  of $\mathfrak{so}(2n+1,\CC)$ on $\Extern\CC^n$.
 
  Note that with these conventions
  \begin{equation*}
    \begin{aligned}
      \pi^{1/2}(X_{\ell,2n+1}) &= \tfrac12\gamma_\ell,\\
      \pi^{1/2}(X_{j\ell}) &= \tfrac12\gamma_j\gamma_\ell,
      \quad 1\le j< \ell\le 2n
      .
    \end{aligned}
  \end{equation*}
  Let
  \begin{equation}
    \label{eq:4}
    E_j = X_{2j-1,2n+1} +\ii  X_{2j,2n+1},
    \quad
    E_{-j} = -X_{2j-1,2n+1} +\ii  X_{2j,2n+1}
    .
  \end{equation}
  Then, from our choice of $\phi$ and $\gamma$, 
  \begin{equation*}
    c_j^\dagger = \pi^{1/2}(E_j) ,
    \quad
    c_j = \pi^{1/2}(E_{-j})
    ,
  \end{equation*}
  coincide with the standard creation annihilation operators
  on $\Gamma_\wedge \mathbb \CC^n$ as in \eqref{creation annihilation j}.
  
  Because of this we see that  $\pi^{1/2}$ is irreducible.
  Indeed, by repeated applications of $\pi^{1/2}(E_j)=a_j$ or
  $\pi^{1/2}(E_{-j})=a_j^\dagger$
  any invariant subspace of $\Extern\CC^n$ under
  $\pi^{1/2}(\mathfrak{so}(2n+1,\CC))$
  has to include all of $\Extern\CC^n$.
  Therefore this representation coincides with the standard $1/2$-spin
  representation of $\mathfrak{so}(2n+1,\CC)$ and  the proof is complete.
\end{proof}
We will need to consider the real Lie algebra $\mathfrak{so}(2n+1)$
alongside $\mathfrak{so}(2n+1,\CC)$. Hence we note the following fact.
\begin{corollary}
  The representation $\pi^{1/2}$  restricts to an irreducible representation,
  also called $1/2$-spin representation, of the real Lie algebra   $\mathfrak{so}(2n+1)$.
\end{corollary}
\begin{proof}
  Indeed, with the same notation as in the proof of the proposition above, we have
  \begin{equation*}
    \begin{aligned}
      \pi^{1/2}(X_{\ell,2n+1}) &= \tfrac12\gamma_\ell,\\
      \pi^{1/2}(X_{j\ell}) &= \tfrac12\gamma_j\gamma_\ell,
      \quad 1\le j< \ell\le 2n
      .
    \end{aligned}
  \end{equation*}
  Being $\gamma(v)$ complex linear in $v\in\CC^{2n}$ we have that
  $\pi^{1/2}$ is an irreducible analytic representation of $\mathfrak{so}(2n+1,\CC)$
  and it naturally restricts to a well defined representation of the real Lie algebra
  $\mathfrak{so}(2n+1)$.
  By the ``Weyl unitary trick'' (cf. \cite[Theorem 3, \SS{1}, Chapter 8, pp. 202-203]{barut_theory_1986})
  the restricted representation is indeed irreducible.
\end{proof}
\begin{remark}\label{remark: lowest weight}
  We note that, under the same conventions as the proof of the proposition above,
  the vector $1\in\Gamma_\wedge\CC^n$ can be seen as a lowest weight vector with relative
  weight $(-\tfrac12,\cdots,-\tfrac12)\in\CC^n$.
  
  To show this let us employ the same notation as in the proof above and
  let us fix a Cartan subalgebra $\mathfrak h$ of $\mathfrak{so}(2n+1,\CC)$ generated by
  the elements
  \begin{equation}
    \label{Cartan subalgebra}
    H_j = \tfrac12 [E_jE_{-j}] = \ii X_{2j-1,2j}
    ,\quad j=1,\dots,n
    .
  \end{equation}
  These generators are normalized in such a way that, if we identify
  $\mathfrak h$ with $\CC^n$ by sending $H_j$ into $e_j$
  and we equip $\CC^n$ with its standard symmetric bilinear form
  $\langle\cdot,\cdot\rangle$, then the dual space $\mathfrak h^\ast$
  is itself isomorphic to $\CC^n$ and the dual of an element $H_j\cong e_j$
  is the element $H_j\cong e_j$ itself.
  Now, $H_j 1 = -\tfrac12 c_j c_j^\dagger 1 = -\tfrac12$.
  Hence the vector $1\in\Gamma_\wedge\CC^n$ is associated to the weight
  $\lambda\in \mathfrak h^\ast$ such that $\lambda(H_j)=-\tfrac12$, for all $j=1,\dots,n$.
  Under the identifications given above, of $\mathfrak h^\ast\cong\CC^n\cong\mathfrak h$,
  where $\CC^n$ is identified with its dual via the standard symmetric bilinear form on $\CC^n$,
  the weight $\lambda$ corresponds to the vector
  \begin{equation*}
    \big(-\tfrac12,\dots,-\tfrac12\big)\in\CC^n
    .
  \end{equation*}
  Under the natural order of $\RR$ this weight is the lowest weight of the representation.
  Hence we have shown that,
  under our conventions,
  $1\in\Gamma_\wedge\CC^n$ is the (normalized) lowest weight vector.

  Similarly, under the same conventions,
  one can show that $e_1\wedge\cdots\wedge e_n\in\Gamma_\wedge\CC^n$
  is the  (normalized) highest weight vector corresponding to the highest weight
  $ \big(\tfrac12,\dots,\tfrac12\big)\in\CC^n$.

  By repeated use of the creation annihilation operators one shows that
  a general weight is of the form 
  $(\underbrace{\pm\tfrac12,\dots,\pm\tfrac12}_{n\text{ times}})$
  with a given number of ``plus signs'' and the complementary number of ``minus signs''.
  From a physical perspective, the plus signs in the \textit{weight} $\lambda$ denote ``filled states'',
  that is to every ``plus sign'' there corresponds a Fermionic particle in the respective state.
\end{remark}
We conclude by stating the following fact which will be regularly used in the following sections.
\begin{proposition}\label{universal rep}
  If we let $\mathfrak U(\mathfrak{so}(2n+1,\CC))$ be the universal enveloping algebra of $\mathfrak{so}(2n+1,\CC)$
  then $\pi^{(1/2)}$ extends naturally to a representation of the universal enveloping algebra and we have
  $\pi^{(1/2)}(\mathfrak U(\mathfrak{so}(2n+1,\CC)) \cong \CCliff(2n)$.
\end{proposition}
\begin{proof}
  This fact follows at once from the universal property of universal enveloping algebras
  (cf. e.g. \cite[Theorem 9.7, p. 247]{hall_lie_2015}).
\end{proof}

\section{Fermions and \texorpdfstring{$C^\infty(\Spin(2n+1))$}{C(Spin(2n+1))}}\label{3}
The proposition \ref{universal rep} expresses the complex Clifford algebra $\CCliff(2n)$
as representation of the universal enveloping algebra of the complex Lie algebra $\mathfrak{so}(2n+1,\CC)$.
In this section, after a short description of the standard representation of the universal enveloping algebra
in terms of certain differential operators acting on a common domain of functions, we describe how to recover from this
(infinite dimensional) representation the $1/2$-spin representation $\pi^{(1/2)}$ of $\mathfrak{so}(2n+1,\CC)$
defined in section \ref{2}.

Consider a connected, simply connected, compact Lie group $\mathbf G$ with Lie algebra $\mathfrak g$.
Let $C^\infty(\mathbf G)$ be the space of complex valued functions on $\mathbf G$.
We denote by $L:\mathbf G\rightarrow\text{End}\left(C^\infty(\mathbf G) \right)$, $L:g\mapsto L_g$,
the action of $\mathbf G$ on $C^\infty(\mathbf G)$ by the \textit{left translation} $L_g$, $g\in G$,
where $L_g f(x)\deq f(g^{-1}x)$, $g,x\in\mathbf G$.
Similarly we denote by $R$ the action of the $\mathbf G$ on $C^\infty(\mathbf G)$ by
the \textit{right translation} $R_g$, $g\in\mathbf G$, where $R_gf(x)\deq f(x g)$, $g,x\in\mathbf G$.

Let us denote by $\mathfrak D(\mathbf G)$ the algebra of differential operators on $C^\infty(\mathbf G)$
generated by the identity and the left-invariant vector fields on $\mathbf G$,
i.e. the vector fields which commute with the left translation.

We have the following important fact
(Cf. \cite[Ch. II, Proposition $1.9$ and its proof, p. 108]{helgason_differential_1979}):
the universal enveloping algebra $\mathfrak U(\mathfrak g)$ is isomorphic
(as a an algebra) to $\mathfrak D(\mathbf G)$.
Moreover, by this isomorphism, the Lie algebra $\mathfrak g$ is represented  on
$C^\infty(\mathbf G)$ 
by the representation $dR:\mathfrak g\rightarrow \mathfrak D(\mathbf G)$
which associates to each element $X\in\mathfrak g$ the corresponding left-invariant vector field $dR(X)$,
where the linear operator $dR(X):C^{\infty}(\mathbf G)\rightarrow C^{\infty}(\mathbf G)$
is defined by
$$
dR(X)f = \frac{d}{dt}R(e^{tX})f|_{t=0},
\quad f\in C^\infty(\mathbf G)
.
$$
\begin{remark}\label{Remark:Clifford-Lie}
  The fact that the universal enveloping algebra is isomorphic (as an algebra)
  to $\mathfrak D(\mathbf G)$ means that the invariant vector fields
  $dR(X_1),$ $\dots,$ $dR(X_n)$ associated with the generators of the Lie algebra 
  $\mathfrak g$
  satisfy the Lie algebra commutation relations, that is\footnote
  {
    We denote by $[XY]$ (no comma)
    the Lie brackets of the Lie algebra $\mathfrak g$
    and by $[A,B]=AB-BA$
    (with comma) the comutator in an associative algebra e.g. $\mathfrak D(\mathbf G)$
    or $\mathfrak U(\mathfrak g)$.
  }
  $[dR(X),dR(Y)]f = dR([XY])f$, $f\in C^\infty(\Spin(2n+1))$, $X,Y\in\mathfrak g$.
  
  Let $\pi^{1/2}=\gamma\circ\phi\circ\kappa$ be the $1/2$-spin representation of $\mathfrak{so}(2n+1,\CC)$
  as given in section \ref{2}.
  With the same notation as in that section we have that
  \begin{equation*}
    \pi^{1/2}(X_{j,2n+1}) = \gamma_j,\quad j=1,\dots,2n
    .
  \end{equation*}
  Note that $\gamma_1,\dots\gamma_{2n}$ satisfy the Lie algebra \textit{commutation} relations
  of $\mathfrak{so}(2n+1)$ and the \textit{anticommutation} relations of the Clifford algebra.
  Now, we can lift any generator $\gamma_j$, $j\in{1,\dots,2n}$, of the Clifford algebra $\CCliff(2n)$ to 
  an invariant (complex) vector field as a differential operator in $\mathfrak D(\mathbf G)$.
  These vector fields will satisfy the commutation
  relations of the Lie algebra $\mathfrak{so}(2n+1,\CC)$
  but \textit{not} the Clifford \textit{anticommutation} relations of the original elements
  $\gamma_1,\dots\gamma_{2n}$.
  To recover the Clifford anticommutation relations we will need to project onto a subspace
  isomorphic to the Fermionic Fock space $\Gamma_\wedge\CC^n$.
  We now turn to the description of this procedure.
\end{remark}
Let $L^2(\mathbf G)$ denote the space of functions from $\mathbf G$ to $\CC$
which are square integrable with respect to the normalized Haar measure $\dd g$ on $\mathbf G$.
By slight abuse of notation we will still denote by $R$ the extension of the representation,
of $\mathbf G$ on $C^\infty(\mathbf G)$
by right translation, to a representation
representation of $\mathbf G$ on $L^2(\mathbf G)$. Note that this extension gives a unitary representation.
We now embed the Fermionic Fock space $\Extern\CC^n$ into $L^2(\Spin(2n+1))$.
\begin{lemma}\label{Lemma 2}
  Let $\pi^{(1/2)}(g)$  denote the 
  $1/2$-spin representation of an element $g\in\Spin(2n+1)$.
  Then the map 
  $$
  F^{1/2}:\Gamma_\wedge\CC^n\hookrightarrow L^2(\Spin(2n+1)),
  \quad F^{1/2}: \psi \mapsto \big( 1, \pi^{(1/2)}(g)\, \psi\big)_{\extp\CC^n}.
  $$
  defines an embedding of the Fermionic Fock space $\Gamma_\wedge\CC^n$ into $L^2(\Spin(2n+1))$.
  Let
  \begin{equation}
    \Psi_0\deq F^{1/2}(1)=\big(1,\pi^{(1/2)}(\cdot)1\big)_{\extp\CC^n},\qquad F_{\Psi_0} \deq \Ran(F^{1/2}),
    \label{F Psi 0}
  \end{equation}
  where $\Ran(F^{1/2})$ denotes the image of $F^{1/2}$.

  The restriction of the right regular representation $R$ of $\Spin(2n+1)$ to $F_{\Psi_0}$
  defines a representation which coincides with the $1/2$-spin representation of $\Spin(2n+1)$.
  Moreover, $F_{\Psi_0} \subset C^\infty(\Spin(2n+1))$ and the restriction
  of $dR$ to $F_{\Psi_0}$ defines a representation of the Lie algebra $\mathfrak{so}(2n+1)$
  which coincides with the $1/2$-spin representation of $\mathfrak{so}(2n+1)$.	
\end{lemma}
\begin{proof}
  Let
  $$
  Y_{(i)j}^\alpha(x) \deq \sqrt d_{\alpha}\, D_{ij}^\alpha (x)\,, 
  \quad i,j=1,\dots,d_{\alpha}\,
  \quad x\in\mathbf G
  ,
  $$
  where $\alpha$ labels an irreducible unitary representation of $\Spin(2n+1)$,
  $d_\alpha$ denotes the dimension of such a representation,
  and $D_{ij}^\alpha(g)$ the $i,j$-matrix element of $g\in\Spin(2n+1)$ in such representation.
  By Peter-Weyl theorem
  (cf. \cite[Chapter 7 \S2, Theorem 1 p.172 and Theorem 2 p.174]{barut_theory_1986}),
  for $\alpha$, $i=1,\dots,d_\alpha$ fixed, the set of functions $(Y_{(i)j}^\alpha)_{j=1,\dots,d_{\alpha}}$
  spans a subspace of dimension $d_\alpha$ which is invariant and irreducible 
  for the \textit{right} regular representation.
  Now take $\alpha=1/2$, and let $fj$, $j=1,\dots,d_{1/2}=2^n$
  be an orthonormal basis of $\Gamma_\wedge\CC^n$ with $f_1=1$.
  Then $Y^{1/2}_{(i)j} = 2^{-n/2}(f_i,\pi^{(1/2)}(g)f_j)_{\Gamma_\wedge\CC^n}$.
  If we pick $i=1$ then by the Peter-Weyl theorem, as described above,
  the set $(Y_{(1)j}^{1/2})_{j=1,\dots,d_{1/2}}$ spans a subspace $\mathcal H_{1/2}$
  which is isomorphic to $\Gamma_\wedge\CC^n$.
  And the isomorphism is indeed the $F^{1/2}$ in the statement of the theorem.
  It is also clear that the right regular representation on $L^2(\Spin(2n+1))$
  restricts on $\mathcal H_{1/2}$ to a representation isomorphic to the $1/2$-spin representation of $\Spin(2n+1)$.
  
  To prove the last part of the statement first note that any $Y_{(i)j}^{\alpha}$, as defined above,
  is smooth, that is $Y_{(i)j}^{\alpha}\in C^\infty(\Spin(2n+1))$
  (for a sketch of the argument cf. e.g.
  \cite[Part I, Chapter 2, Appendix to section 2.]{gelfand_representations_2012}).
  Hence $dL$ is well defined on $\mathcal H_{1/2}$ which is by definition the image of $F^{1/2}$.
  By definition of $dL$ it is also clear that $dL$, restricted to $\mathcal H_{1/2}$,
  realizes a representation of the Lie algebra $\mathfrak{so}(2n+1)$
  isomorphic to the $1/2$-spin representation. The proof is therefore complete.
\end{proof}
\begin{corollary}
  The representation $dR$ of $\mathfrak{so}(2n+1)$ extends to a representation $dR_{\CC}$
  of $\mathfrak{U}(\mathfrak{so}(2n+1,\CC))$ on $C^\infty(\Spin(2n+1))$,
  where, as before, $\mathfrak{so}(2n+1,\CC)$
  denotes the complexification of the Lie algebra $\mathfrak{so}(2n+1)$.
\end{corollary}
\begin{proof}
  The representation $dR$ of $\mathfrak{U}(\mathfrak{so}(2n+1))$
  associates to every element $X\in\mathfrak{U}(\mathfrak{so}(2n+1))$ a differential operator
  acting on the complex space $C^\infty(\Spin(2n+1))$.
  Hence the complex-linear extension of $dR$ is well defined on $C^\infty(\Spin(2n+1))$
  and gives a representation $dR_\CC$
  of $\mathfrak{U}(\mathfrak{so}(2n+1,\CC))$ isomorphic to the $1/2$-spin representation.
\end{proof}

\section{Time evolution of a Fermionic state}\label{quasi-Hamiltonian}\label{4}
For quantum mechanical applications it is not enough to consider an algebra of differential operators
on $C^\infty(\Spin(2n+1))$. For example, to discuss the time evolution of the system,
it is also necessary to consider the operators as unbounded operators in the Hilbert space $L^2(\mathbf G)$.
In particular the natural question is whether an operator initially defined on $C^\infty(\mathbf G)$
defines a unique unbounded operator on $L^2(\mathbf G)$.
The main objective of this section is to show that we have a well defined notion of ``quasi-Hamiltonian'',
which lifts the notion of the Hamiltonian for a system of Fermions, to an unbounded,
essentially selfadjoint, positive operator on
$L^2(\Spin(2n+1))$ with domain $C^\infty(\Spin(2n+1))$.
We begin with some general considerations.

Let $\mathfrak g$ be a real semisimple Lie algebra. Let $\theta:\mathfrak g\rightarrow\mathfrak g$
be one of the equivalent Cartan involutions on $\mathfrak g$.
In the case where $\mathfrak g$ is the Lie algebra of a compact semisimple Lie group we take
$\theta$ to be the identity.
Let
\begin{equation}
  X^\ast = -\theta(X),\quad X\in\mathfrak g
  .
  \label{UEA and involution}
\end{equation} 
We extend this involution to
an antilinear involution on $\mathfrak U(\mathfrak g_{\CC})$, where $\mathfrak g_\CC$ is the
complexification of $\mathfrak g$. This operation makes $\mathfrak U(\mathfrak g_{\CC})$ into a $\ast$-algebra.
An element $X\in\mathfrak U(\mathfrak g_\CC)$ is said to be \textit{Hermitian} (as an element of the universal enveloping algebra) when $X=X^\ast$.
  
Consider the algebra $\mathfrak D(\mathbf G)$ of left-invariant smooth differential operators
in $L^2(\mathbf G)$ with common invariant domain $C^\infty(\mathbf G)$
and let $\mathfrak D_\CC(\mathbf G)\deq\mathfrak D(\mathbf G)\otimes_\RR\CC$
denote its complexification.
On $\mathfrak D_\CC(\mathbf G)$ we have an antilinear involution,
which we also denote by $\ast$, which sends the unbounded operator $D\in\mathfrak D_\CC\mathbf G)$ 
to its Hilbert-adjoint $D^\ast$ with respect to the scalar product in $L^2(\mathbf G)$.
This involution makes $\mathfrak D(\mathbf G)$ into a $\ast$-algebra.
Consider the representation $dR$ of the universal enveloping algebra $\mathfrak U(\mathfrak g_\CC)$.
On $C^\infty(\mathbf G)$ we have indeed that $dR(X)^\ast=-dR(X)$, $X\in\mathfrak g$.
But if we consider $dR(X)$ as unbounded operator on $L^2(\mathbf G)$ with domain $C^\infty(\mathbf G)$
then the domain of $dR(X)^\ast$ will in general be larger than the domain of $dR(X)$,
that is, for $X\in\mathfrak{g}$, the operator $dR(\ii X)=\ii dR(X)$ is in general Hermitian but not selfadjoint\footnote
  {
    In the context of unbounded operators in a Hilbert space, an operator $T$ with domain $\Dom(T)$ is \textit{Hermitian} when it satisfies $\Dom(T)\subset\Dom(T^\ast)$
    and $T|_{\Dom(T)}=T^\ast_{\Dom(T)}$.
    The operator $T$ is selfadjoint when in addition the stronger condition $\Dom(T)=\Dom(T^\ast)$ holds.
    In the algebraic context of universal enveloping algebras, 
    an element $X\in\mathfrak U(\mathfrak g_\CC)$ is said to be Hermitian when $X=X^\ast$, where $X^\ast$
    is in the sense of \eqref{UEA and involution}.
    These two, in general different, concepts for an object to be Hermitian coincides
    when we identify the universal enveloping algebra $\mathfrak U(\mathfrak g_\CC)$ 
    with the algebra $\mathfrak D(\mathbf G)$ of smooth right-invariant vector fields  acting on
    $C^\infty(\mathbf G)$.
  }.
Therefore we cannot say that $dR(\ii X)^\ast=dR(\ii X)$ holds when we picture $dR(\ii X)$ as unbounded operators
on $L^2(\mathbf G)$ with domain $C^\infty(\mathbf G)$.
One could try to extend the operator $dR(\ii X)$ to a selfadjoint operator by enlarging its domain.
This might be possible for one operator $dR(X)$ for a fixed $X\in\mathfrak U(\mathfrak g)$.
But, since different $X,Y\in\mathfrak U(\mathfrak g)$ are elements of an \textit{algebra} of operators,
we need to have a \textit{common} invariant domain of definition for both $dR(X)$ and $dR(Y)$.
Hence, in general one cannot expect to find an extension of $dR$
which sends Hermitian elements of $\mathfrak g_\CC$ to self-adjoint operators in $L^2(\mathbf G)$
with common domain of selfadjointness.
One could argue that this requirement is too strong and not necessarily the most natural.
Perhaps a more natural situation, which is obtained in the context of compact semisimple Lie groups,
is the following
(cf. e.g. \cite[Corollary 10.2.10, p.270]{schmudgen_unbounded_1990}).	
Let $\mathbf G$ be a \textit{compact} Lie group with Lie algebra $\mathfrak g$.
Then
\begin{equation}
  \overline{dR(X^\ast)} = dR(X)^\ast,\quad X\in\mathfrak U(\mathfrak g_\CC),
  \label{eq:essential star iso}
\end{equation}
where the over-line on the left hand side denotes the operator closure.
Note that this property implies in particular that any Hermitian element $D_\CC\in\mathfrak D(\mathbf G)$
is automatically \textit{essentially selfadjoint}\footnote
{
  That is, it admits a unique extension to a selfadjoint operator.
}.

We now turn to the notion of commuting unbounded operators.
There are two natural notions of commuting unbounded operators, weakly commuting and strongly commuting.
We give the precise definitions.

Given two unbounded operators $A,B$ with common domain $\mathcal D$ in a Hilbert space $\mathfrak H$,
we say that $A,B$ \textit{weakly commute} on $\mathcal D$ when $ABv=BAv$ for all $v\in\mathcal D$.
Given two \textit{selfadjoint} unbounded operators $A,B$ we say that $A,B$ \textit{strongly commute}
when $e^{\ii t A}e^{\ii s B} = e^{\ii s B} e^{\ii t A}$ for all $s,t\in\RR$, where $e^{\ii tC}$
denotes the unitary one-parameter group generated by a selfadjoint operator $C$ (cf. \cite[Theorem VIII.13]{reed_i_1981}
for a justification of this definition).

Regarding the relation between strong and weak commutativity of operators on a Hilbert space we have the following result due to Nelson.
\begin{lemma}[{\cite[Corollary 9.2]{nelson_analytic_1959}}]\label{Nelson}
	Let $A,B$ be two Hermitian unbounded operators on a Hilbert space
	$\mathcal H$ and let $\mathcal Q$ be a dense linear subspace of $\mathcal H$ such that $\mathcal Q$ is contained in the domain of $A$,
	$B$, $A^2$, $AB$, $BA$, and $B^2$, and such that 
	$A,B$ \textit{weakly commute} on $\mathcal Q$.
	If the restriction of $A^2+B^2$ to $\mathcal Q$ is essentially selfadjoint then $A$ and $B$ are essentially selfadjoint and
	their closures $\overline A$, $\overline B$ \textit{strongly commute}.\qed
\end{lemma}
A direct consequence of this lemma are the following facts, which will be used in this section and the following.
\begin{proposition}\label{strongly commuting}
  Let $\mathbf G$ be a compact Lie group with Lie algebra $\mathfrak g$.
  Let $\mathfrak g_\CC$ be the complexified Lie algebra of $\mathfrak g$, and
  $\mathfrak U(\mathfrak g_\CC)$ its universal enveloping algebra.
  Let $X,Y\in \mathfrak U(\mathfrak g)$ be two \textit{commuting} operators
  (in the algebraic sense of elements in the universal enveloping algebra).
  Then 
  \begin{enumerate} 
  \item the closed operators $\overline{dR(X)},\overline{dR(Y)}\in\mathfrak D(\mathbf G)$ \textit{strongly commute};
  \item if $dR(X)$ is positive (semi-)definite, and $dR(Y)$ is Hermitian,
    then
    $$
    \exp(-\overline{dR(X)}) \exp(\ii \overline{dR(Y)}) = \exp(\ii \overline{dR(Y)}) \exp(-\overline{dR(X)})
    ,
    $$
    where we recall that $\overline{dR(X)}$ and $\overline{dR(Y)}$
    are the \textit{unique} closed extensions of $dR(x)$, respectively $dR(Y)$,
    and $\overline{dR(X)}>0$.
  \end{enumerate}
\end{proposition}
\begin{proof}
  The statement of point \textit{1} follows from  \eqref{eq:essential star iso} and Nelson's Lemma  \ref{Nelson}.
  Indeed, if $X,Y$ commute in the universal enveloping algebra
  then $dR(X)$ and $dR(Y)$ weakly commute on $C^\infty(\mathbf G)$ because $dR$ is a representation of
  $\mathfrak U(\mathfrak g)$ with domain $C^\infty(\mathbf G)$.
  Now, for $\mathbf G$ a compact group, equation \eqref{eq:essential star iso}
  tells us that \textit{any} Hermitian element in the algebra $\mathfrak D(\mathbf G)$ is essentially self adjoint on
  $C^\infty(\mathbf G)\subset L^2(\mathbf G)$. 
  Therefore in particular, for any $X,Y\in\mathfrak U(\mathfrak g_\CC)$, we have that the operators
  $dR(X)$, $dR(X)^2=dR(X^2)$, $dR(X)dR(Y)=dR(XY)$, $dR(X)+dR(Y)=dR(X+Y)$
  have the same domain $C^\infty(\mathbf G)$, and are essentially selfadjoint there.
  Hence the hypothesis of the Lemma \ref{Nelson} are satisfied with
  $A=dR(X)$ and $B=dR(Y)$ and statement \textit{1} follows.
  The statement of point \textit{2} is a straightforward application of spectral calculus
  (cf \cite[Section VIII.5]{reed_i_1981}).
\end{proof}
\begin{remark}\label{remark:commutation}
  Because of the above proposition we only need to check whether two operators commute
  as elements of the universal enveloping algebra. From the Proposition \ref{strongly commuting} 
  it then follows automatically that
  their closures are \textit{selfadjoint} and \textit{strongly commuting}.
\end{remark}

With this proposition we have completed the considerations from the general theory.
We can now turn to the application that we have in mind.
\begin{definition}[Quasi-Fermionic vector fields]\label{quasi fields}
  Let $X_{ij}$ ,$i,j=1,\dots,2n+1$, be the standard basis (cf. \eqref{Lie brakets})
  of the Lie algebra $\mathfrak{so}(2n+1)$ of
  the Lie group $\Spin(2n+1)$.
  Let us denote by $D_{ij} \deq dR(X_{ij})$ the corresponding left-invariant vector fields on $\Spin(2n+1)$.
  We define the following operators (cf. \eqref{eq:4})
  \begin{align*}
    D_{k}^+ &\deq D_{2k-1,2n+1} + \ii D_{2k,2n+1}\,, \\
    D_{k}^- &\deq -D_{2k-1,2n+1} + \ii D_{2k,2n+1}\,,\qquad k=1,\dots,n\,
              ,
  \end{align*}
  as linear operators on $C^\infty(\Spin(2n+1),\CC)\subset L^2(\Spin(2n+1))$.
\end{definition}
\begin{definition}[Quasi-Hamiltonian operator]\label{Quasi-Hamiltonian}
  Let us fix $n$ strictly positive numbers 
  $E_1,\dots,E_n$, with $0<E_1\le\dots\le E_n$.
  Using for $D^{\pm}_k$ the notation of the previous paragraph we call a
  \textit{quasi}-\textit{Hamiltonian} the operator
  $$
  \tilde{H} = \sum_{k=1}^n E_k D^+_k D^-_k\,
  ,
  $$
  acting on $C^\infty(\Spin(2n+1))$.
\end{definition}
\begin{remark}
  The operators $D^\pm_k$
  restricted to the finite dimensional subspace $F_{\Psi_0}\subset C^\infty(\Spin(2n+1))$,
  given in \eqref{F Psi 0} of \ref{Lemma 2},
  satisfy the canonical anticommutation relations.
  For this reason we call these operators ``quasi-Fermionic''.
  Similarly we named the operator $\tilde{H}$ ``quasi-Hamiltonian'' because, restricted to the subspace $F_{\Psi_0}\cong\Extern\CC^n$,
  it coincides with the \textit{free Fermionic Hamiltonian operator}
  $H \deq \sum_k E_k c_k^\dag c_k$, where the creation-annihilation operators $c^\dagger_k,c_k$, $k=1,\dots,n$,
  were defined in section~\ref{2}.
\end{remark}
\begin{theorem}\label{Theorem:quasi Hamiltonian}
  The unbounded operator $\tilde{H}$ with domain $C^\infty(\Spin(2n+1))$ in 
  $L^2(\Spin(2n+1))$ is a positive, essentially selfadjoint operator.
  Moreover the quasi-Hamiltonian can be decomposed on $C^\infty(\Spin(2n+1))$ as
  $$
  \tilde{H} =  P_0 +\ii B_0 ,\qquad P_0 \deq -\sum_{k=1}^n E_k L_k ,\quad B_0\deq \sum_{k=1}^n E_k T_k ,
  $$
  with $T_k\deq D_{2k-1,2k}$, $L_k\deq D_{2k-1,2n+1}^2+D_{2k,2n+1}^2$, $k=1,\dots,n$,
  and the following properties are satisfied:
  \begin{enumerate}
  \item The operator $P_0$ and $\ii B_0$,
    with domains $\Dom(P_0)$, $\Dom(\ii B_0)$ both equal to $C^\infty(\Spin(2n+1))$,
    are essentially selfadjoint in $L^2(\Spin(2n+1))$.
    Moreover $P_0$ is positive definite.
    In particular $-P_0$ and $B_0$ are closable and their closures $-\overline{P_0}, \overline{B_0}$
    are selfadjoint operators which generate respectively a semigroup and a unitary group
    (we consider $-P_0$
    because generators of semigroups are usually taken to be negative definite).
  \item The operators $\overline{P_0}$ and $\overline{\ii B_0}$ strongly commute.
    The operators $T_k, L_k$, $k=1,\dots,n$, are essentially selfadjoint. The unique selfadjoint closures
    $\overline{T_k}$, $k=1,\dots,n$ define a family of strongly commuting unbounded operators.
    Moreover each $\overline{T_k}$, $k=1,\dots,n$ strongly commutes with each $\overline{L_\ell}$, $\ell=1,\dots,n$.
    In particular $\overline{\ii B_0}$ strongly commutes with $\overline{P_0}$.
  \end{enumerate}
\end{theorem}
\begin{proof}
  First note that $\tilde{H}$ is well defined on $C^\infty(\Spin(2n+1))$,
  since $D^\pm$ are linear combinations of smooth vector fields, in particular $D^-$ maps
  $C^\infty(\Spin(2n+1))$ into $C^\infty(\Spin(2n+1))$
  (indeed $\mathfrak D(\Spin(2n+1)$ is an algebra).
  Using the above definition of the operators $D_k^+,D_k^-$ in terms of the operators $D_{ij}$  in
  Definition \ref{quasi fields} we have, on $C^\infty(\Spin(2n+1))$,
  \begin{align*}
    &\tilde{H} = -\sum_{k=1}^n E_k(D_{2k-1,2n+1} + \ii D_{2k,2n+1})(D_{2k-1,2n+1} - \ii D_{2k,2n+1})\\
    &= -\sum_{k=1}^n E_k(D_{2k-1,2n+1})^2 - \!\sum_{k=1}^n E_k(D_{2k,2n+1})^2
                - \!\ii\sum_{k=1}^n E_k\,[D_{2k-1,2n+1},D_{2k,2n+1}] \\
              &= -\sum_{k=1}^{n} E_k 
                \left( (D_{2k-1,2n+1})^2 + (D_{2k,2n+1})^2 \right) - \ii\sum_{k=1}^n E_k D_{2k-1,2k}
	.
  \end{align*}
  Therefore we obtain
  \begin{equation}\label{eq:Hamiltonian E prime}
    \tilde{H} = P_0 + \ii B_0
  \end{equation}
  where, $P_0$ and $\ii B_0$ are defined in the statement of the theorem.
  Note that by \eqref{eq:essential star iso} all Hermitian operators we are handling are essentially selfadjoint.
  Moreover the operators $D_k^+D_k^-$, $k=1,\dots,n$, are positive definite since $D_k^+$
  is the formal adjoint of $D_k^-$.
  This implies that the quasi-Hamiltonian $\tilde{H}$ is  essentially selfadjoint
  (by \eqref{eq:essential star iso}) and is positive definite (since is the sum of positive definite operators).
  This concludes the proof of the first part of the theorem.

  Property \textit{1} is proved by a similar argument.
  The fact that $P_0$ is closable and its closure defines a semigroup follows from the fact that
  $P_0$ is essentially selfadjoint (therefore closable) and positive definite (hence defines a semigroup).
  Similarly $B_0$ is closable because $\ii B_0$ is essentially selfadjoint and therefore defines a unitary
  one-parameter group.

  We now turn to the proof of property \textit{2}.
  By \eqref{eq:essential star iso} we obtain that $\ii T_k=\ii D_{2k-1,2k}$ and 
  $L_k$ are self-adjoint. 
  Note that, by remark \ref{remark:commutation} if two elements $X,Y$
  in the universal enveloping algebra $\mathfrak{U}(\mathfrak{so}(2n+1))$ commute,
  then their representation $dR(X)$, $dR(X)$ admit closures $\overline{dR(X)}$, $\overline{dR(Y)}$
  which strongly commute.
  Hence to prove the commutation properties of point \textit{2},
  it is enough to perform the computation on the universal enveloping algebra.

  Now, from the fact
  that the elements $X_{2k-1,2k}$, $k=1,\dots,n$
  generate a maximal commutative subalgebra (Cartan subalgebra)\footnote
  {
    Cf. remark \ref{remark: lowest weight}
  } of the Lie algebra 
  $\mathfrak{so}(2n+1)$ we obtain that
  the operators $T_k= D_{2k-1,2k} = dR(X_{2k-1,2k})$, $k=1,\dots,n$ form a commuting family of operators.
  
  As above, by remark \ref{remark:commutation},  to show that 
  $L_\ell$ commutes with $X_{2k-1,2k}$,
  for all $\ell,k\in \{1,\dots,n\}$,
  it is enough to prove that the corresponding elements of the universal enveloping algebra commute.
  Consider
  $$
  L_\ell^{\mathfrak U} \deq (X_{2\ell-1,2n+1})^2 + (X_{2\ell,2n+1})^2,
  \quad
  \ell=1,\dots,n
  ,
  $$
  be the element, associated to $L_k$, in the universal enveloping algebra of $\mathfrak{so}(2n+1)$.
  It is enough to prove that 
  \begin{equation}
  [L^{\mathfrak U}_\ell, X_{2k-1,2k}]=0, \quad\text{for all }\ell,k=1,\dots,n
  .\label{eq:1}
  \end{equation}
  This follows from the following straightforward computations.
  Using the identity $[X^2,Y]=X[X,Y]+[X,Y]X$ for any $X,Y\in \mathfrak U(\mathfrak{so}(2n+1)_\CC)$ 
  we get
  \begin{multline}
    [L^{\mathfrak U}_\ell,X_{2k-1,2k}] =\\
    = X_{2\ell-1,2n+1}[X_{2\ell-1,2n+1}, X_{2k-1,2k}] 
    + [X_{2\ell-1,2n+1}, X_{2k-1,2k}]X_{2\ell-1,2n+1} +\\
    +  X_{2\ell,2n+1}[X_{2\ell,2n+1}, X_{2k-1,2k}] 
    + [X_{2\ell,2n+1}, X_{2k-1,2k}]X_{2\ell,2n+1}
    .
  \end{multline}
  Using in this expression the commutation relations \eqref{eq:4} of the standard basis of $\mathfrak{so}(2n+1)$, 
  we obtain for $\ell,k=1,\dots,n$
  \begin{multline*}
    [L^{\mathfrak U}_\ell,X_{2k-1,2k}] =\\
    = -X_{2\ell-1,2n-1} \,\delta_{2\ell-1,2k-1}X_{2n+1,2k}
      -\delta_{2\ell-1,2k-1}X_{2n+1,2k} \,X_{2\ell-1,2n+1} +\\
    + X_{2\ell,2n+1}\, \delta_{2\ell,2k}X_{2n+1,2k-1} +
      \delta_{2\ell,2k} X_{2n+1,2k-1} \,X_{2\ell,2n+1}
      .
  \end{multline*}
  Now, using in this expression the fact that $X_{ij}=-X_{ji}$ for all $1\le i<j\le 2n+1$, and 
  collecting the Kronecker deltas into a unique Kronecker delta which multiplies everything, we get
  \begin{multline*}
    [L^{\mathfrak U}_\ell,X_{2k-1,2k}] 
    = \delta_{k,\ell} \, (X_{2\ell-1,2n+1} \,X_{2k,2n+1}
      + X_{2k,2n+1} \,X_{2\ell-1,2n+1}\\
    - X_{2\ell,2n+1}\,X_{2k-1,2n+1} -
      X_{2k-1,2n+1} \,X_{2\ell,2n+1} )
      .
  \end{multline*}
  Finally, using the identity $\delta_{ij}f(i,j)=\delta_{ij}f(i,i)$ where $f(i,j)$ 
  is any function of $i,j\in\NN$,
  we get
  \begin{align*}
    [L^{\mathfrak U}_\ell,X_{2k-1,2k}]
    &= \delta_{k,\ell} \,
      (X_{2k-1,2n+1} \,X_{2k,2n+1}
      + X_{2k,2n+1} \,X_{2k-1,2n+1}\\
    &\qquad\qquad\quad- X_{2k,2n+1}\,X_{2k-1,2n+1} -
      X_{2k-1,2n+1} \,X_{2k,2n+1} )\\
    &=0
      .
  \end{align*}
  This proves \eqref{eq:1}. 
  As a consequence it is now clear that
  $P_0$ commutes with $B_0$ which concludes the proof of property \textit{2} and of the theorem.
\end{proof}

\section{Relation with stochastic processes}\label{5}
From theorem \ref{Theorem:quasi Hamiltonian} we have that the quasi-Hamiltonian is
$$
\tilde{H} = P_0 + \ii B_0,
$$
with $P_0\deq\sum_{k=1}^n E_k L_k$ and $B_0\deq\sum_{k=1}^n E_k T_k$, where all the operators
are defined on $C^\infty(\Spin(2n+1))$.

Since the operator $B_0$ appears in $\tilde{H}$ multiplied by the imaginary unit $\ii$
we cannot associate directly to the closure $\overline{\tilde{H}}$ a (real) stochastic process
on $\Spin(2n+1)$.
For this reason we consider, together with $P_0$, $B_0$, and $\tilde{H}$ above,
the following operator
\begin{equation}
  P \deq P_0 + B_0,
  \quad \Dom(P)\deq C^\infty(\Spin(2n+1))
  .\label{P0,P}
\end{equation}

It turns out that it is possible to associate a stochastic diffusion processes on $\Spin(2n+1)$ to both closures
$\overline{P_0}$ and $\overline{P}$ in $L^2(\Spin(2n+1)$.
First we see that both $-\overline{P_0}$ and $-\overline{P}$ generate probability semigroups in the following sense.
\begin{lemma}
  The operators $-P$, respectively $-P_0$ are essentially selfadjoint on $C^\infty(\Spin(2n+1))$
  and their closures $-\overline P$, $-\overline{P_0}$
  are infinitesimal generators of strongly continuous semigroups
  which act on $L^{2}(\Spin(2n+1))$ as 
  convolution semigroups of probability measures with support on $\Spin(2n+1)$.
\end{lemma}
\begin{proof}
The statement follows from \cite[Theorem 3.1]{jorgensen_representations_1975}.
\end{proof}
Now we characterize the stochastic processes generated by $-\overline{P_0}$ and $-\overline P$ in terms of the SDEs these processes satisfy.
Before doing so let us briefly introduce the notions of stochastic differential equation (SDE)
on a manifold and of generator of a diffusion process (cf. e.g. \cite{ikeda_stochastic_1989}). 
  
Let $\mathcal M$ be a connected smooth manifold of dimension $d$.
Moreover for convenience let us assume $\mathcal M$ to be compact.
This assumption simplifies somewhat the discussion and is sufficient for our purposes
because we will in the sequel only deal with manifolds associated to compact Lie groups.
In particular if $\mathcal M$ is a compact manifold, then every $C^\infty$-vector field on it 
is complete, that is, the flow associated to the given vector field can be extended
to all times. This allows us to define a stochastic process globally on the manifold $\mathcal M$
(without the need of the introduction of an explosion time).

Let us denote by $\mathcal X(\mathcal M)$ the set of $C^\infty$-vector fields on $\mathcal M$.
Let us consider $r\in\NN$ vector fields $A_0,A_1,\dots,A_r\in\mathcal X(\mathcal M)$ on 
$\mathcal M$. 

Let $(\Omega,(\mathscr F_t)_{0\le t<\infty},\mathbb P)$ be a filtered probability space;
we denote by $(W(t)) =(W^1(t),\dots,W^r(t))$
 an $r$-dimensional $\mathscr F_t$-adapted Brownian motion starting at zero, $W(0)=0$.
Finally, let $\xi$ be an $\mathscr F_0$-measurable $\mathcal M$-valued random variable.

Consider now an $\mathscr F_t$-adapted stochastic process $X=X(t)$ on $\mathcal M$,
that is an $\mathscr F_t$-adapted random variable $X=(X(t))$ 
with values in the continuous functions $C^0([0,\infty);\mathcal M )$.
Contrary to the previous sections, in this section the letter $X$ will be reserved to denote a random variable.

Suppose that for every $f\in C^\infty(\mathcal M)$ the stochastic process $X=(X(t))$ 
satisfies $\mathbb P$-almost surely the following integral equation
\begin{equation}\label{eq:integral SDE}
  f(X(t))-f(\xi) = \int_0^t \sum_{k=1}^r (A_k f)(X(s))\circ dW^k(s) + \int_0^t (A_0 f)(X(s))\dd s
  , 
\end{equation}
for all\footnote
{
  By saying that the equality holds $\mathbb P$-almost surely, for all $t$,
  we are saying that the right hand side and the left hand side define \textit{indistinguishable} processes.
} 
$t\in [0,\infty)$, where $\circ dB$ denotes integration in the Stratonovich sense (see, e.g. \cite{ikeda_stochastic_1989}).
Then we will say that the $\mathcal M$-valued stochastic process $X=(X(t))$ is a 
\textit{solution} to \eqref{eq:integral SDE}.

Let us spend few words on the notion of \textit{strong solution} regardless whether we are
on a manifold $\mathcal M$ or just in $\RR^d$.
Given a notion of solution it is natural to ask whether it satisfies some given 
\textit{initial condition} which we take here to be
a point $x\in\mathcal M$ (withoug any randomness). 
A solution to \eqref{eq:integral SDE} with (non random) initial conditions $\xi=x$,
is then a stochastic process $X_x$ starting at time $0$ at $x$.

More preciselly, we are asking for a function 
$F:\mathcal M\times C^0([0,\infty); \RR^r)\rightarrow C^0([0,\infty); \mathcal M)$
which maps the initial condition $x\in\mathcal M$ and the given realization of the Brownian motion 
$W=(W(t))$ into a realization of a process $X_x=(X_x(t))$ on the manifold $\mathcal M$.
Moreover $F$ is such that $X_x=F(x,W)$ is a solution
to \eqref{eq:integral SDE} with initial condition $\xi=x$ with probability one
and with given Brownian motion $W=(W(t))$.
Since at some point we would like to integrate $X_x$ both
with respect to $x\in\mathcal M$ and with respect to $\PP$
it is natural to ask that $F$ be jointly measurable in $x$ and $W=(W(t))$.
It turns out that this is \textit{not} always possible.
When it is, $X_x=F(x,W)$ is called a \textit{strong solution} to \eqref{eq:integral SDE}
with initial condition $\xi=x\in\mathcal M$ with probability one
(cf. the discussion in \cite[Section V.10]{rogers_diffusions_1994} and \cite[Chapter IV, section 1, esp. pp.162-163]{ikeda_stochastic_1989}).

In the context of smooth manifolds the situation is particularly good
because we are considering SDE with smooth coefficients.
Indeed one has a result (cf. \cite[Chapter V, Section 1., Theorem 1.1, p.249]{ikeda_stochastic_1989})
which states that given an initial condition $x\in\mathcal M$ and an $r$-dimensional Brownian motion $W=(W(t))$, then a strong solution to $\eqref{eq:integral SDE}$ always exists and is unique\footnote{
  The idea behind this result is that the manifold $\mathcal M$
  is locally diffeomorphic to $\RR^d$ where $d$ is the dimension of the manifold $\mathcal M$.
  This means that locally the SDE \eqref{eq:SDE} (and hence \eqref{eq:integral SDE}) can be written in coordinates as a standard SDE on $\RR^d$.
  One can apply standard results about existence and uniqueness of solutions of SDEs to these
  local realizations.
  Finally one needs to patch together different local solutions into a global solution.
  Details can be found in the above mentioned \cite{ikeda_stochastic_1989}.
}.

Once this important detail about how the initial condition is understood 
we can give meaning to the following shorthand, which we shall refer to as a
\textit{Stratonovich SDE} on the (compact) manifold $\mathcal M$:
\begin{equation}\label{eq:SDE}
  \begin{cases}
    dX(t) = \sum_{k=1}^{r} A_k (X(t))\circ dW^k(t) + A_0(X(t))dt\,,\\
    X(0) = x\,,\qquad x\in\mathcal M.
  \end{cases}
\end{equation}
The meaning associated to \eqref{eq:SDE} is that we consider a strong solution $X$ of
\eqref{eq:integral SDE} (with initial conditions $\xi=x$ with probability one)
and then define a solution
to \eqref{eq:SDE} to be the random
variable $X_x = F(x,W)$, where $F$ is the map which defines our  strong solution $X$.

We now define the notion of stochastic diffusion process and of its generator.

First consider a more general case.
For $x\in\mathcal M$, let $X_x$ be a continuous stochastic process adapted to a filtration $\mathscr F_t$
in the probability space $(\Omega,\mathscr F,\mathbb P)$.
For simplicity we consider a stochastic process defined for all $t\in [0,\infty)$
and with values in the space of continuous maps $[0,\infty)\rightarrow\mathcal M$ (where $\mathcal M$ is always assumed to be compact)
such that $X(0)=x$ (where equality means $\mathbb P$-a.s.).

Let $P_x$ be the probability law associated to the random variable $(X_x(t))$.
This means that $P_x$ is the image measure under the measurable mapping $X_x = (X_x(t))$ of the probability measure $\mathbb P$.
Assume that $x\mapsto P_x$ is universally measurable\footnote
{
  Cf., e.g. \cite[p.1]{ikeda_stochastic_1989}.
} and that $P_x$ is uniquely determined by $x\in\mathcal M$.\footnote
{
  These conditions are actually automatically satisfied when $X_x$
  is the strong solution to \eqref{eq:SDE}.
}
Moreover assume that there exists a linear operator $\mathfrak L$ with domain $\Dom(\mathfrak L)$
in $C(\mathcal M)$,
such that for every $f\in\Dom(\mathfrak L)$,
$$
X_f(t) \deq f(X(t))-f(X(0)) - \int_0^t (\mathfrak L f)(X(s))\dd s
$$	
is a martingale with continuous sample paths 
and adapted to the filtration $\mathscr F_t$ associated to 
$X_x(t)$ (cf. \cite[Chapter IV, Theorem 5.2, p.207]{ikeda_stochastic_1989}).
Then the family of probability measures $(P_x)_{x\in\mathcal M}$ is called a \textit{diffusion} \textit{generated} by the operator $\mathfrak L$.

When, for every $x\in\mathcal M$, $X_x$ is the stochastic diffusion process on the manifold $\mathcal M$
which is the strong solution to \eqref{eq:SDE} with initial condition $X(0)=x$,
then we have the following representation \cite[Chapter V, Theorem 1.2, p.253]{ikeda_stochastic_1989}.
The family of probability laws $(P_x)_{x\in\mathcal M}$, associated with the strong solutions $X_x$ to \eqref{eq:SDE} with initial conditions $x\in\mathcal M$, is a diffusion generated
by the operator 
$$
\mathfrak L \deq \frac{1}{2}\sum_{j=1}^r A_k(A_k f) + A_0 f,
\quad f\in C^\infty(\mathcal M)
, 
$$
(where, as before, the manifold $\mathcal M$ is assumed to be compact)
and the vector fields $A_0,A_1,\dots,A_r\in \mathcal X(\mathcal M)$
are interpreted as differential operators with common domain $C^\infty(\mathcal M)$.

We now go back to our setting where the manifold $\mathcal M=\Spin(2n+1)$
and collect the specialized version of the results recalled above.
Doing so we give the characterization of the generators $-\overline{P_0}$ and $-\overline{P}$
(given by \eqref{P0,P}) in terms of stochastic processes on $\Spin(2n+1)$.

\begin{remark}[Notation]
  In this section we do not distinguish between an element
  $X_{ij}$ in the  Lie algebra and the corresponding differential operator $D_{ij}=dR(X_{ij})$
  (cf. \ref{quasi fields}).
  In particular, depending on the context, we identify $A_k$, $k=1,\dots,2n$, with either 
  $D_{2n+1,k}$  or $X_{2n+1,k}$. 
  Similarly, the differential operator $B_0$ in theorem \ref{Theorem:quasi Hamiltonian}
  will be considered also as a vector field without changing notation.
\end{remark}
\begin{lemma}[Stochastic processes associated to $P_0$ and $P$]\label{SDEs and P P_0}
  The following statements hold.
  \begin{enumerate}
  \item The Stratonovich SDEs  on $\Spin(2n+1)$ 
    \begin{align*}
      (P)\quad
      &\begin{cases}
        dY(t) = \sum_{k=1}^{2n} \sqrt{E'_k}A_k (Y(t))\circ dW^k(t) - B_0(Y(t))dt\,,\\
        Y(0) = x\,,\qquad x\in\Spin(2n+1)
      \end{cases}\\
      (P_0)\quad
      &\begin{cases}
        dX(t) = \sum_{k=1}^{2n} \sqrt{E'_k}A_k (X(t))\circ dW^k(t) \\
        X(0) = x\,,\qquad x\in\Spin(2n+1)
        ,
      \end{cases}
    \end{align*}
    with
    $E'_{2k+1}\deq E'_{2k} \deq E_k$, $k=1,\dots,n$,
    and $(W^k(t)$, $k=1,\dots,2n)$, a standard Brownian motion in $\RR^{2n}$, 
    are well defined and admit a \textit{unique} \textit{strong solution}.
    
  \item The operators $-\overline P$ and $-\overline{P}_0$ (acting on $L^2(\Spin(2n+1))$  are
    the generators of the \textit{diffusion processes} given by the strong solutions of $(P)$, $(P_0)$ respectively.
  \end{enumerate}
\end{lemma}
\begin{proof}
  For the first statement see \cite[Chapter 5, Theorem 1.1 p.249]{ikeda_stochastic_1989}.
  The second statement follows from \cite[Theorem 1.2, p.253]{ikeda_stochastic_1989}.
\end{proof}
The following result relates the time evolution semigroup generated by the quasi-Hamiltonian
$\tilde{H}$ in $L^2(\Spin(2n+1))$ with
a stochastic diffusion process on $\Spin(2n+1)$
generated by the second order part in $-\tilde{H}$.
\begin{theorem}\label{Feynam-Kac} 
  We have the following representations of the semigroup generated by the closure $-\overline{\tilde{H}}$
  of $-\tilde{H}$
  \begin{equation}
    (f, e^{-t \overline{ \tilde{H}} } g)_{L^2(\Spin(2n+1))}
    =
    \Expect_X\left[ \overline{f(0)} \, \left( e^{ -\ii\, t \,\overline{B_0} } g\right)\big(X(t)\big)  \right]
    ,
    \qquad t\ge0
    ,
  \end{equation}
  where $\Expect_X$ denotes the expectation with respect to the process generated by $-P_0$,
  $\overline{B_0}$, respectively $\overline{\tilde{H}}$, denotes the closure
  (which exists by theorem \ref{Theorem:quasi Hamiltonian}) of the operator $B_0$,
  respectively $\tilde{H}$;
  $\overline{f(0)}$ denotes complex conjugation,
  and $f,g\in C({\Spin(2n+1)})\subset L^2(\Spin(2n+1))$.
\end{theorem}
\begin{proof}
  First note that $e^{-t \overline{ \tilde{H}} }$ is a bounded operator for all $t\in\RR^+$.
  Hence $f,g$ can be taken in $L^2(\Spin(2n+1))$.
  The equality follows directly from the representation of the Hamiltonian as $\tilde{H} = P_0 + \ii B_0$,
  the fact that $\overline{P_0}$ and $B_0$ strongly commute, and the Markov property of
  the semigroup generated by $-\overline{P_0}$
  (which is a consequence of point \textit{2} of lemma \ref{SDEs and P P_0}):
  \begin{align*}
    (f, e^{-t \overline{ \tilde{H}} } g)_{L^2(\Spin(2n+1))}
    &=
      (f, e^{-t\, (  \overline{P_0} + \ii \overline{B_0} )}  g)_{L^2(\Spin(2n+1))}\\
    &=
      (f, e^{-t\, \overline{ P_0} } e^{-\ii\,t\, \overline{ B_0 }} g)_{L^2(\Spin(2n+1))}\\
    &=
      \Expect_X\left[ \overline{f(0)} \, \left( e^{ -\ii\, t \,\overline{B_0} } g\right)\big(X(t)\big)  \right]
      .\qedhere
  \end{align*}
\end{proof}

\section*{Acknowledgments}
  I wish to thank Prof.\ Disertori for her support during this project
  which constitutes part of the research carried out during my Ph.D.
  I would also like to thank Prof.\ Albeverio for inspiring conversations concerning topics related
  to this project.
  Part of this research was founded by DFG via the grant AL 214/50-1
  ``Invariant measures for SPDEs and Asymptotics''.

\end{document}